\documentclass[11pt]{amsart}
\usepackage{graphicx,amssymb,amsmath,amsthm}
\usepackage{amsfonts}
\usepackage{enumerate}
\usepackage{dsfont}
\usepackage{cite}
\usepackage[colorlinks, citecolor=red]{hyperref}
\usepackage{mathrsfs}
\usepackage{epsfig}
\usepackage{lscape}
\usepackage{subfigure}
\usepackage{epstopdf}
\usepackage{caption}
\usepackage{algorithm}
\usepackage{algpseudocode}
\usepackage{multirow}
\usepackage{geometry}
\usepackage{paralist}
\usepackage{enumerate}
\usepackage{url}

\newcommand{\cK}{{\mathcal K}}
\newcommand{\cN}{{\mathcal N}}
\newcommand{\cH}{{\mathcal H}}
\newcommand{\cW}{{\mathcal W}}
\newcommand{\cQ}{{\mathcal Q}}
\newcommand{\cR}{{\mathbb R}}

\newtheorem{definition}{Definition}[section]

\newtheorem{theorem}[definition]{Theorem}
\newtheorem{lemma}[definition]{Lemma}

\newtheorem{remark}[definition]{Remark}

\newtheorem{assumption}{Assumption}[section]
\date{}

\begin{document}
\baselineskip 18pt
\bibliographystyle{plain}
\title[A symmetric AM algorithm for TV minimization]{A symmetric alternating minimization algorithm for total variation minimization}

\author{Yuan Lei}
\address{School of Mathematics, Hunan University, Changsha 410082, China.}
\email{yleimath@hnu.edu.cn}

\author{Jiaxin Xie}
\address{School of Mathematical Sciences, Beihang University, Beijing, 100191, China }
\email{xiejx@buaa.edu.cn}

\begin{abstract}
In this paper, we propose a novel symmetric alternating minimization algorithm to solve a broad class of total variation (TV) regularization problems. Unlike the usual $z^k\to x^k$ Gauss-Seidel cycle, the proposed algorithm performs the special $\overline{x}^{k}\to z^k\to x^k$ cycle. The main idea for our setting is the recent symmetric Gauss-Seidel (sGS) technique which is developed for solving the multi-block convex composite problem.
This idea also enables us to build the equivalence between the proposed method and the well-known accelerated proximal gradient (APG) method.
The faster convergence rate of the proposed algorithm can be directly
obtained from the APG framework and numerical results including image denoising, image deblurring, and analysis sparse recovery problem demonstrate the effectiveness of the new algorithm.
\end{abstract}

\maketitle

\section{Introduction}
Although more modern and excellent techniques which are specifically tailored to image processing have been developed,
total variation (TV) regularization is still a widely used method in the community of applied mathematics and engineering, for its good properties for preserving contours and sharp edges  in objects with spatial structures; see for instance \cite{herrmann2019,hong2019,liuzx2019,rof,tian2019,sardy2019,scherzer2010,zhang-2018,zhangb2019}.
In this paper, we focus on solving a broad class of  TV regularization problems, where the problem shall be reformulated as an unconstrained problem using the penalty decomposition approach \cite{chenxj2016,luzs2015} and then be efficiently solved in the framework of alternating minimization (AM) \cite{cai2019,chambolle,guminov2019,nesterov2015,saha,sunabcd,suntao2019,yang-g,yang-c}.

The last decades witnessed excellent progress and exciting development regarding the efficient algorithms for TV minimization problem, like smoothing-based methods \cite{rof,vogel96,vogel}, iterative shrinkage thresholding (IST) algorithms \cite{beck-fast, beck-fast-ieee,twist}, interior point methods \cite{ngl1}, primal-dual strategies \cite{pda,tian2019}, augmented Lagrangian methods \cite{liuzx2019,wuchunlin}, to mention just a few. One may refer to \cite{Esser2010,scherzer2010,yang-2017} for a brief overview on this topic. Among them, FTVd \cite{yang-g,yang-c} and the alternating direction method of multipliers (ADMM)-based algorithms \cite{salsa,csalsa,tval3,taomin}, also known as a special split Bregman algorithm (SBA) \cite{bregman}, have been widely used.
FTVd is actually an AM-based method. The authors \cite{yang-g,yang-c} first introduced an auxiliary vector to transfer the problem to an unconstrained problem and then solved it by the AM algorithm. Therefore, FTVd inherits the computational simplicity and efficiently of AM, and performs much better than a number of existing methods such as the lagged diffusivity algorithms \cite{chanld,vogel96}, and some Fourier and wavelet shrinkage methods \cite{deelamani}. Unlike FTVd, ADMM-based methods actually deal with a constrained model and have the characteristics of robustness with respect to the regularization parameter \cite{taomin}.

The AM optimization algorithms have been widely known for a long time \cite{chambolle,Ortega1970,saha}. In recent decades the AM algorithms  have become an important kind of method in both convex optimization and engineerings, such as machine learning problems \cite{Andresen2016}, phase retrieval \cite{cai2019}, sparse recovery \cite{xie2018}, and semidefinite programming \cite{sunabcd}. Sublinear $1/k$ convergence rate of the AM algorithms was proved in \cite[Lemma 5.2]{beck-AM}. Despite the same convergence rate as for the (proximal) gradient method, AM algorithms converge faster in practice as they are free of the
choice of the step-size and are adaptive to the local smoothness of the problem. At the same time, there are accelerated proximal gradient (APG) methods that use a momentum term to have a faster convergence rate of $1/k^2$ \cite{Bai2018,beck-fast,sun-14-apg}. Ideally, one wants to find an algorithm that can not only keep the simplicity and efficiently of AM but also share the faster convergence rate.
Recently, a new novel symmetric Gauss-Seidel (sGS) technique was developed by Li, Sun and Toh \cite{li-sgs} for solving the multi-block convex composite optimization problem. This technique is playing increasingly important roles in solving very large problems (e.g., \cite{chenliang,qsdpnal,li-sgs,xiao2018,WangXiao2019}), as it can break down the large scale problems into several smaller ones and solve correspondingly by making full use of its favorable structures.
 In this paper, we extend this technique to the AM-based algorithm. A benefit of the extension is that one can introduce a momentum term to the AM algorithm, and obtain an accelerated symmetric AM method for TV minimization.
The contributions of this paper are as follows.
\begin{itemize}
  \item We propose a symmetric AM algorithm for solving the TV regularization problem. Our algorithm takes the special $\overline{x}^{k}\to z^k\to x^k$ (sGS) iterative scheme rather than the usual Gauss-Seidel $z^k\to x^k$ iterative scheme. 
      Note that in every iteration the algorithm needs to compute the $x$ twice, which could lead to poor performances if updating $x$ is its dominant computational cost, even it shares a faster convergence rate.
      We show that the computation of $\overline{x}^{k}$ can be very  simple and efficient, meaning that the proposed algorithm can keep the simplicity  of the classical AM algorithm.
  \item We show that the proposed algorithm is actually equivalent to the APG method and consequently a faster convergence rate can be easily obtained from the APG framework. In addition, we  show that the proposed algorithm can obtain the $\epsilon$-optimal solution within $O(1/\epsilon^{1.5})$ iterations.
  \item  Numerical examples show the good performance of our proposed algorithm for image denoising, image deblurring, and analysis sparse recovery.
\end{itemize}

This paper is organized as follows. After introducing some preliminaries in Section \ref{sect.2}, we present our algorithm in Section \ref{sect.3}. In Section \ref{sect.4}, we analyze the convergence properties of the proposed algorithm. In Section \ref{sect.6}, numerical experiments and comparisons with existing  solvers are carried out to confirm the effectiveness of the proposed methods. We conclude the paper in Section \ref{sect.7}.

\section{Preliminaries}
\label{sect.2}
\subsection{TV model}
Consider the following ill-posed linear inverse problem
$$
f=\cK x+\eta,
$$
where $\cK$ is a linear operator such as a blur convolution, $\eta$ is additive noise and $f$ is a degradation.
Without loss of generality, we consider images with a square $n$-by-$n$ domain, and treat it as a column vector $x:=(x_1,x_2,\ldots,x_{n^2})^T$.
The TV model, which is a widely used generic minimization model to recover $x$ from $f$ and $\cK$, consists of solving the following lasso-like problem
\begin{equation}
\label{TV}
\min_{x}\bigg\{\Phi(x):=\sum_{i=1}^{n^2}||D_ix||_2+\frac{\mu}{2}||\cK x-f||_2^2\bigg\},
\end{equation}
where $\mu>0$ is a regularization parameter, and $D_i x\in\cR^2$ represents the first-order finite difference of $x$ at pixel $i$ in both horizontal and vertical directions. For $i=1,\ldots,n^2$,
$D_ix:=((D^{(1)}x)_i,(D^{(2)}x)_i)^T\in\cR^2$, where
$$
(D^{(1)}x)_i:=\left\{
                \begin{array}{ll}
                  x_{i+n}-x_i, & \hbox{if $1\leq i\leq n(n-1)$,} \\
                  x_{i-n(n-1)}-x_i, & \hbox{otherwise,}
                \end{array}
              \right.
$$
$$
(D^{(2)}x)_i:=\left\{
                \begin{array}{ll}
                  x_{i+1}-x_i, & \hbox{if $i=n,2n,\ldots,n^2$,} \\
                  x_{i-n+1}-x_i, & \hbox{otherwise.}
                \end{array}
              \right.
$$
For simplicity of notation, we denote $D=((D^{(1)})^T,(D^{(2)})^T)^T$.
We need the following assumption for the convergence analysis of the proposed algorithm.
\begin{assumption}
\label{assumptiom}
$\cN(\cK)\cap\cN(D)=\{0\}$, where $\cN(\cdot)$ represents the null space of a matrix.
\end{assumption}

\subsection{An alternating minimization algorithm}
\label{sect.2-2}

The alternating minimization (AM) method based on the classical quadratic penalty approach was first introduced for solving the TV-based image restoration problem by
\cite[Algorithm 1]{yang-g}, see also \cite{baizj,xiao2012,yang-c,yang-l1}. An auxiliary $z$ is introduced in \eqref{TV} to give the following decomposition transformation
\begin{equation}
\label{TV-f}
\min_{x,z}\bigg\{\Psi(x,z):=\sum_{i=1}^{n^2}||z_i||_2+\frac{\beta}{2} \sum_{i=1}^{n^2}||z_i-D_ix||^2_2+\frac{\mu}{2}||\cK x-f||_2^2\bigg\}.
\end{equation}
Here at each pixel we use $z_i=((z_1)_i,(z_2)_i)^T\in\cR^2$ to approximate $D_ix=(\big (D^{(1)}x)_i, (D^{(2)}x)_i\big)^T\in\cR^2, i=1,\ldots,n^2$ and $\beta>0$ is a penalty parameter.

For a fixed $x$, the first two terms in \eqref{TV-f} are separable with respect to $z_i$, so minimizing \eqref{TV-f} for $z$ is equivalent to solving, for $i=1,2,\ldots,n^2$,
\begin{equation}
\label{z-sub1}
\min_{z_i}\bigg\{||z_i||_2+\frac{\beta}{2}\|z_i-D_ix\|^2_2\bigg\}.
\end{equation}
According to \cite{yang-g,yang-c}, the solution of \eqref{z-sub1} is given explicitly by the two-dimensional shrinkage
\begin{equation}
\label{z-sub}
z_i=\max\bigg\{\|D_ix\|_2-\frac{1}{\beta},0\bigg\}\frac{D_ix}{\|D_ix\|_2}, i=1,2,\ldots,n^2,
\end{equation}
where $0\cdot(0/0)=0$ is assumed. On the other hand, for a fixed $z$, \eqref{TV-f} is quadratic in $x$, and the minimizer $x$ is given by the following normal equations
\begin{equation}
\label{x-sub}
\bigg( D^TD+\frac{\mu}{\beta}\cK^T\cK  \bigg)x=D^Tz+\frac{\mu}{\beta}\cK^Tf.
\end{equation}
 The iterative procedure of AM for solving \eqref{TV-f} with a fixed $\beta$ is shown below.

\bigskip
\centerline{
\fbox{
\parbox{0.98\textwidth}{
{\bf Algorithm 1: An AM method for solving problem \eqref{TV-f}.}\\
Input $f,\cK,\mu>0$, and $\beta>0$. Initialize $x^0=f$. For $k=0,1,\ldots$, iteratively compute:
\begin{compactenum}[\bf Step 1.]
\item  Compute $z^{k+1}$ according to \eqref{z-sub} with $x=x^k$.
\item  Compute $x^{k+1}$ by solving \eqref{x-sub} with $z=z^{k+1}$.
\end{compactenum}
}}}
\bigskip
\begin{remark}
\label{remark_AM}
By Lemma $5.2$ in \cite{beck-AM}, one can establish the $O(1/k)$ convergence rate for Algorithm $1$. That is, let $\{(x^k,z^k)\}$ be the sequence generated by Algorithm $1$, then for any $k\geq 1$,
$$
\Psi(x^k,z^k)-\Psi^*\leq O(\frac{1}{k}),
$$
where $\Psi^*$ is the optimal value of problem \eqref{TV-f}.
\end{remark}

\subsection{The accelerated proximal gradient method}
\label{sect.2-3}

Consider the following general convex optimization model:
\begin{equation}
\label{gen-pro}
\min_{u\in\cR^p}\big\{F(u):=p(u)+q(u)\big\}.
\end{equation}
Here, $p:\cR^m\to(-\infty,+\infty]$ is an extended-valued, proper, closed and convex function (possible nonsmooth); $q:\cR^m\to\cR$ is convex and continuously differentiable with Lipschitz continuous gradient.
Given any symmetric positive definite matrix $\cH$, define $\omega(\cdot,\cdot):\cR^m\times \cR^m\to \cR$ by
$$
\omega_{\cH}(u,w):=q(w)+ \langle \nabla q(w),u-w\rangle+\frac{1}{2}\langle u-w,\cH(u-w)\rangle.
$$
Then for any $ u,w\in \cR^m$, there exists a symmetric positive definite matrix $\cH$ such that
\begin{equation}
\label{inequ}
q(u)\leq \omega_{\cH}(u,w).
\end{equation}
For any $\widehat{u}^1=u^0\in \cR^m$ and $t_1=1$, the $k$-th iteration of the accelerated proximal gradient (APG) method \cite{beck-fast,sun-14-apg} can be reformulated as
\begin{equation}
\label{apg}
\left\{
\begin{array}{l}
u^{k}:=\arg\min_{u\in\cR^m}\big\{p(u)+\omega_{\cH}(u,\widehat{u}^k)\big\},
\\[2mm]
t_{k+1}:=\frac{1+\sqrt{1+4t_k^2}}{2},
\\[2mm]
\widehat{u}^{k+1}:=u^{k}+\frac{t_k-1}{t_{k+1}}(u^{k}-u^{k-1}).
\end{array}
\right.
\end{equation}
Let $\mathcal{H}$ be an self-adjoint positive semidefinite matrix, we define $\|u\|^2_{\mathcal{H}}=\langle u,\mathcal{H}u\rangle$. The APG method has the following $O(1/k^2)$ convergence rate.
\begin{theorem}[\cite{sun-14-apg}, Theorem 2.1]
\label{theorem_apg}
Suppose that the sequence $\{u^k\}$ is generated by \eqref{apg}. Then
$$
F(u^k)-F(u^*)\leq \frac{2\|u^0-u^*\|^2_{\cH}}{(k+1)^2},
$$
where $u^*$ is an optimal solution to \eqref{gen-pro}.
\end{theorem}

\section{A symmetric alternating minimization algorithm}
\label{sect.3}
In this section, we develop the symmetric alternating minimization algorithm for solving \eqref{TV-f}
and Algorithm $2$ lists the basic iterations of our algorithm. One of the attractive feature of Algorithm $2$ is that it takes the special $\overline{x}^{k}\to z^k\to x^k$ iterative scheme rather than that of AM taken the usual Gauss-Seidel $z^k\to x^k$ iterative scheme. By using this special iterative scheme, we can establish the $O(1/k^2)$  rate of convergence for Algorithm $2$  (see Theorem \ref{remark_apg}).

\bigskip
\centerline{
\fbox{
\parbox{0.98\textwidth}{
{\bf Algorithm 2: A symmetric alternating minimization algorithm for solving problem \eqref{TV-f}.}\\
Input $f,\cK,\mu>0$, and $\beta>0$. Initialize $z^0=D^Tf,\widehat{z}^1=z^0$ and $t_1=1$. For $k=1,2,\ldots$, iteratively execute the following steps:
\begin{compactenum}[\bf Step 1.]
\item  Compute $\overline{x}^{k}$ by solving \eqref{x-sub} with $z=\widehat{z}^{k}$.
\item  Compute $z^{k}$ according to \eqref{z-sub} with $x=\overline{x}^k$.
\item  Compute $x^{k}$ by solving \eqref{x-sub} with $z=z^{k}$.
\item  Set $t_{k+1}=\frac{1+\sqrt{1+4t_k^2}}{2}$, and compute $\widehat{z}^{k+1}=z^{k}+\frac{t_k-1}{t_{k+1}}(z^{k}-z^{k-1})$.
\end{compactenum}
}}}
\bigskip

One may note that Algorithm $2$ has to solve the normal equations \eqref{x-sub} twice per-iteration. That is, the computational cost of Algorithm $2$ will be twice as much as that of Algorithm $1$. Consequently, Algorithm $2$ may be much slower than Algorithm $1$, even though it shares a faster convergent rate (see Theorem \ref{remark_apg}). Fortunately, in following we shall show that  Algorithm $2$ actually has to solve only one normal equations \eqref{x-sub} per-iteration, i.e., the step $1$ in Algorithm $2$ to compute $\overline{x}^{k}$ can be very simple and efficient.
For the convenience of subsequent analysis, we denote
$$\cW :=D^TD+\frac{\mu}{\beta}\cK^T\cK, $$
$b:=\frac{\mu}{\beta}\cK^Tf$ and $\tau_k:=\frac{t_k-1}{t_{k+1}}$. It follows from Assumption \ref{assumptiom} that $\cW$ is nonsingular.

 Note that when $k\geq 2$, the step 1 in Algorithm $2$ for computing $\overline{x}^k$ can be obtained by solving the linear system
\begin{equation}
\label{remk-1}
\cW \overline{x}^k=D^T\widehat{z}^k+b=D^T(z^{k-1}+\tau_{k-1}(z^{k-1}-z^{k-2}))+b.
\end{equation}
From the step 3 in Algorithm $2$ and \eqref{x-sub} we know that $D^T z^k=\cW x^k-b$. As a result, \eqref{remk-1} can be rewritten as
$$
\cW \overline{x}^k= \cW x^{k-1}+\tau_{k-1} \cW(x^{k-1}-x^{k-2}).
$$
Thus the step 1 in Algorithm $2$ for calculating $\overline{x}^k$ can be obtained by a much simpler form
\begin{equation}
\label{daam_1}
\overline{x}^k=\left\{
                \begin{array}{ll}
                  \cW^{-1}(D^T \widehat{z}^k+b), & \hbox{if $k=1,2$;} \\
                  x^{k-1}+\tau_{k-1} (x^{k-1}-x^{k-2}), & \hbox{otherwise.}
                \end{array}
              \right.
\end{equation}
Therefore, when $k>2$, the main cost per-iteration of Algorithm $1$ is in step 3 for solving the normal equation \eqref{x-sub}.

\section{Convergence analysis}
\label{sect.4}
In this section, we establish the convergence results for Algorithm $2$. The following theorem shows that Algorithm $2$ shares a $O(1/k^2)$ rate of convergence.

\begin{theorem}
\label{remark_apg}
Let $(x_{\beta}^*,z_{\beta}^*)$ be an optimal solution to problem \eqref{TV-f} and set \begin{equation}
\label{Q}
\mathcal{Q}=I+D\mathcal{W}^{-1}D^T.
\end{equation}
Suppose that $\{(x^k,z^k)\}$ is the sequence generated by Algorithm $2$. Then for any $k\geq 1$,
\begin{equation}
\label{sgs_apg_2}
\Psi(x^k,z^k)-\Psi(x_{\beta}^*,z_{\beta}^*)\leq \frac{2\beta\|z^0-z_{\beta}^*\|_{\mathcal{Q}}^2}{(k+1)^2}.
\end{equation}
\end{theorem}

By utilizing Theorem \ref{remark_apg}, we can establish a convergence result for problem \eqref{TV}, which shows that Algorithm $2$ can obtain an $\epsilon$-optimal solution for problem \eqref{TV} within $O(1/\epsilon^{1.5})$ iterations.
\begin{theorem}
\label{theorem_Daam}
 Let $\epsilon>0$ and let $\{(x^k,z^k)\}$ be the sequence generated by Algorithm $2$ from any initial point $z^0$ with  $\beta$ chosen as
\begin{equation}
\label{theorem_Daam_2}
\beta=\frac{32C}{\epsilon^2}.
\end{equation}
Suppose $x^*$ is an optimal solution of problem \eqref{TV} and the sequence $\{\Psi(x^k,z^k)\}$ is bounded above, i.e., $\Psi(x^k,z^k)\leq C$, where $C>0$ is a constant.
Then an $\epsilon$-optimal solution of \eqref{TV}, i.e., $\Phi(x^k)-\Phi(x^*)\leq \epsilon$, can be obtained by Algorithm $2$  after at most
\begin{equation}
\label{theorem_Daam_1}
K: =\max\bigg\{\frac{16\sqrt{C}\|z^0-z_{\beta}^*\|_{2}}{\epsilon^{1.5}}-1,1\bigg\}
\end{equation}
iterations.
Here we use $(x_{\beta}^*,z_{\beta}^*)$ to denote an optimal solution to problem \eqref{TV-f}.
\end{theorem}

\begin{remark}
\label{remk-daam}
The condition that $\{\Psi(x^k,z^k)\}$ is bounded above appears to be fairly weak. But unfortunately we are not able to prove this condition. As is noted in the above section, Algorithm $2$ and the APG method are identical in some sense.
Therefore, monotone APG, also known as monotone FISTA \cite{beck-fast-ieee}, can yield the boundness of $\{\Psi(x^k,z^k)\}$.
\end{remark}

\subsection{Proof of Theorem \ref{remark_apg}}
To state conveniently, we will denote $h(z):=\sum_{i=1}^{n^2}\|z_i\|_2$ and $g(x):=\frac{\mu}{2}\|\cK x-f\|_2^2$, then problem \eqref{TV} becomes
\begin{equation}
\label{TVg}
\min_{x}\big\{\Phi(x)=h(Dx)+g(x)\big\}
\end{equation}
and problem \eqref{TV-f} becomes
\begin{equation}
\label{TVd}
\min_{x,z}\bigg\{\Psi(x,z)=h(z)+g(x)+\frac{\beta}{2}\|z-Dx\|^2_2\bigg\}.
\end{equation}

\begin{proof}[Proof of Theorem \ref{remark_apg}]
Firstly, inspired by \cite{chenliang,xiao2018,qsdpnal,li-sgs}, we will show that Algorithm $2$ is equivalent to the APG method \eqref{apg}. Specifically, Algorithm $2$ can be regarded as the APG method applied to problem \eqref{TVd} with $u:=(x,z)$, $p(u):=h(z)$, $q(u):=g(x)+\frac{\beta}{2}\|z-Dx\|^2_2$ and $\cH$ be chosen as
\begin{equation}
\label{h}
\cH:=\beta\left( \begin{array}{cc}
                     \cW & -D^T \\
                      -D & I+D\cW^{-1}D^T \\
                   \end{array}
           \right).
\end{equation}
Note that $g(x)=\frac{\mu}{2}\|\cK x-f\|_2^2$. Then the APG method with $\cH$ chosen as \eqref{h}  for solving problem  \eqref{TVd} takes the following iterations
\begin{equation}
\label{apg_DAAM}
(x^k,z^k):=\arg\min_{x,z}\bigg\{h(z)+g(x)+\frac{\beta}{2}\|z-Dx\|^2_2+\frac{\beta}{2}\| z-\widehat{z}^k\|_{D\cW^{-1}D^T}^2\bigg\}.
\end{equation}

On the other hand, consider the sequences $\{(x^k,z^k)\}$ generated by Algorithm $2$, from the step 1 in Algorithm $2$ we can get
$
\mathcal{W} \overline{x}^k=D^T\widehat{z}^k+b
$. Thus
\begin{equation}
\label{sgs1}
\overline{x}^k=\cW^{-1}(D^T\widehat{z}^k+b).
\end{equation}
By the step 3 in Algorithm $2$ we can obtain
\begin{equation}
\label{sgs2}
\cW x^k=D^T z^k+b,
\end{equation}
and from the step 2 in Algorithm $2$ we know that
\begin{eqnarray}
\label{sgs3}
0&\in&\partial h(z^k)+\beta(z^k-D\overline{x}^k)\nonumber
\\
 &=& \partial h(z^k)+\beta(z^k-D\cW^{-1}(D^T\widehat{z}^k+b))\nonumber
\\
&=& \partial h(z^k)+\beta(z^k-Dx^k)+\beta D\cW^{-1}D^T(z^k-\widehat{z}^k),
\end{eqnarray}
where the second and third equalities are due to \eqref{sgs1} and \eqref{sgs2}, respectively.
Note that \eqref{sgs2} together with \eqref{sgs3} are the first order optimality conditions to problem \eqref{apg_DAAM}. Therefore, for solving problem \eqref{TVd}, the sequences generated by Algorithm $2$  and APG are identical.
Hence, by Theorem \ref{theorem_apg} we know that the sequence $\{u^k:=(x^k,z^k)\}$ generated by Algorithm $2$ satisfies
\begin{equation}
\label{sgs_apg_1}
\Psi(x^k,z^k)-\Psi(x_{\beta}^*,z_{\beta}^*)\leq \frac{2\|u^0-u^*\|^2_{\mathcal{H}}}{(k+1)^2},
\end{equation}
where $u^*:=(x_{\beta}^*,z_{\beta}^*)$ is an optimal solution to problem \eqref{TVd} and $u^0=(x^0,z^0)$ is the initial point.
Note that in Algorithm $2$, the initial point $u^0$ is only relevant to $z^0$ and any choice of the $x^0$ is acceptable. Hence
\begin{equation}
\Psi(x^k,z^k)-\Psi(x_{\beta}^*,z_{\beta}^*)\leq \inf_{x^0\in \mathbb{R}^{n^2}}\frac{2\|u^0-u^*\|^2_{\mathcal{H}}}{(k+1)^2}
=\frac{2\beta\|z^0-z_{\beta}^*\|_{\mathcal{Q}}^2}{(k+1)^2}.
\end{equation}
This completes the proof of Theorem \ref{remark_apg}.
\end{proof}

\subsection{Proof of Theorem \ref{theorem_Daam}}

We begin with the following result.
\begin{lemma}[\cite{tanzhao}, Lemma \uppercase\expandafter{\romannumeral3}.1]
\label{lemma_d}
Let $\Psi^*$ be the optimal value of problem \eqref{TVd} and $\Phi^*$ be the optimal value of problem \eqref{TVg}. Then $\Psi^*\leq\Phi^*$.
\end{lemma}

We next establish the following lemma to bound the difference $\Phi(x^k) -\Phi(x^*)$.

\begin{lemma}
\label{lemma_d2}
Let $\{(x^k,z^k)\}$ be the sequence generated by Algorithm $2$ and $u^*=(x_{\beta}^*,z_{\beta}^*)$ denote the optimal solution of \eqref{TV-f} or \eqref{TVd}. If the sequence $\{\Psi(x^k,z^k)\}$ is bounded above, i.e., $\Psi(x^k,z^k)\leq C$, where $C>0$ is a constant, then
\begin{equation}
\label{lemma_d_1}
\Phi(x^k) -\Phi(x^*)\leq\frac{4\beta\|z^0-z_{\beta}^*\|_{2}^2}{(k+1)^2}+\sqrt{\frac{2C}{\beta}}.
\end{equation}
Here $x^*$ denotes the optimal solution to \eqref{TV} or \eqref{TVg} and $z^0$ is an arbitrary point.
\end{lemma}
\begin{proof}
The proof follows mainly from the ideas in \cite{beck-smooth,tanzhao}. Since the  sequence $\{\Psi(x^k,z^k)\}$ is bounded from above, i.e.,
$$
h(z^k)+g(x^k)+\frac{\beta}{2}\|z^k-Dx^k\|^2_2\leq  C,
$$
therefore, it follows that
$$
\frac{\beta}{2}\|z^k-Dx^k\|^2_2\leq C,
$$
so
\begin{equation}
\label{proof_d_2}
\|z^k-Dx^k\|_2\leq\sqrt{\frac{2C}{\beta}}.
\end{equation}
Using Lemma \ref{lemma_d} and \eqref{sgs_apg_2}, we have
\begin{equation}
\label{proof_d_3}
\Psi(x^k,z^k)-\Phi(x^*)\leq \frac{2\beta\|z^0-z_{\beta}^*\|_{\cQ}^2}{(k+1)^2},
\end{equation}
where $\cQ$ is defined as \eqref{Q} and $z^0$ is a initial point.
We therefore conclude that
\begin{eqnarray}
\label{proof_d_4}
\Phi(x^k)&=& h(Dx^k)+g(x^k)\nonumber
\\
&=& h(z^k)+g(x^k)+ h(Dx^k)-h(z^k)\nonumber
\\
&\leq&\Psi(x^k,z^k)+\|z^k-Dx^k\|_2\nonumber
\\
&\leq& \Phi(x^*)+\frac{2\beta\|z^0-z_{\beta}^*\|_{\cQ}^2}{(k+1)^2}+\sqrt{\frac{2C}{\beta}}.
\end{eqnarray}
The first inequality follows from the function $h$ is Lipschitz continuous with parameter $L_h=1$, the second inequality is obtained from \eqref{proof_d_2} and \eqref{proof_d_3}. Using the fact that $\|D\cW^{-1}D^T\|_2\leq 1$ together with the inequality \eqref{proof_d_4}, we know \eqref{lemma_d_1} holds. This completes the proof.
\end{proof}

Now we are ready to prove Theorem \ref{theorem_Daam}.

\begin{proof}[Proof of Theorem \ref{theorem_Daam}]
Take $\beta=\frac{32C}{\epsilon^2}$, by Lemma \ref{lemma_d2} we have
$$
\Phi(x^k) -\Phi(x^*)\leq\frac{128C\|z^0-z_{\beta}^*\|_{2}^2}{(k+1)^2\epsilon^2}+\frac{\epsilon}{4}.
$$
To guarantee the inequality $\Phi(x^k) -\Phi(x^*)\leq\epsilon$, it is sufficient that $\frac{128C\|z^0-z_{\beta}^*\|_{2}^2}{(k+1)^2\epsilon^2}\leq \frac{3\epsilon}{4}$ holds. The inequality is satisfied if condition \eqref{theorem_Daam_1} holds, this completes the proof.
\end{proof}

\section{Numerical experiments}
\label{sect.6}
In this section, we present numerical results to demonstrate the effectiveness of Algorithm $2$  and compare it with some state-of-the-art codes that are available and applicable to lasso-like optimization problem with TV regularization. All the computational tasks are implemented by running {\sc Matlab} R2018b in a PC configured with Intel(R) Core(TM) I5-8500 @3.00GHz CPU and 8GB RAM.
To make it easier to compare across different algorithms, we use one uniform stopping criterion for all the algorithms we tested, that is,
\begin{equation}
\label{stop-rule}
\frac{\|x^{k+1}-x^k\|_2}{\max\{1,\|x^k\|_2\}}<\zeta,
\end{equation}
where $\zeta$ denotes the error tolerance.

\subsection{Implementation details}
We test both grayscale and color images in our experiments;
see Figure \ref{fig:test},
\begin{figure}
\centering
\includegraphics[width=0.95\linewidth]{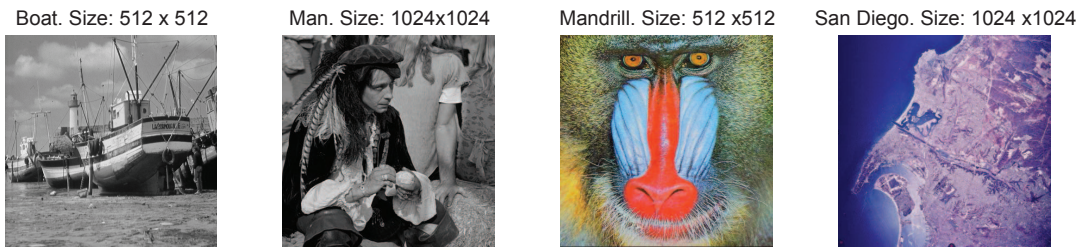}
\caption{The test images used in this paper.}
\label{fig:test}
\end{figure}
where the image sizes are given and they are available as TIFF files\footnote{\url{http://sipi.usc.edu/database/}}.
During our implementation, all the pixels of the original images are first scaled into the range between $0$ and $1$.

In our experiments, several kinds of blurring kernels including Gaussian, motion and average are tested.
All blurring effects are generated using the {\sc Matlab} function {\tt imfilter} with periodic boundary conditions.
To state conveniently, we denote (G, {\tt size}, {\tt sigma}), (M, {\tt len}, {\tt theta}) and (A, {\tt size}) as the Gaussian blur, the motion blur and the average blur, respectively.
For the RGB color images,
we combine various blurring kernels above to generate some within-channel blurs and cross-channel blurs.
In all tests, the additive noise used is Gaussian noise with
 zero mean and various deviations. For simplicity, we denote the  standard deviation as $\sigma$.

It is well known that the solution of \eqref{TV} converges to that of \eqref{TV-f} as $\beta\to \infty$.
However, when $\beta$ is large, our algorithms converge slowly. Thanks to the huge numerical experimentation in \cite{yang-g,yang-c}, we can safely choose $\beta=2^7$ which is a reasonably large value such that the SNR value of the recovered images stabilize.
In addition,  inspired by \cite{yang-g,yang-c}, in our experiment,  we set $\mu=0.05/\sigma^2$ which is based on the observation that $\mu$ should be inversely proportional to the noise variance, while the constant $0.05$ is determined empirically so that the restored images have reasonable SNR and relative errors, here we mention that SNR denotes the signal noise ration.

\subsection{Verifying the acceleration effectiveness of Algorithm $2$}
In Figure \ref{fig:am}, we depict the evolution of the SNR improvement with respect to the number of  iterations and the CPU time for the image man to illustrate the fast convergence properties of Algorithm $2$. Clearly Algorithm $2$ converges faster than the basic AM method (Algorithm $1$), as expected.

\begin{figure}
\centering
\includegraphics[width=0.48\linewidth]{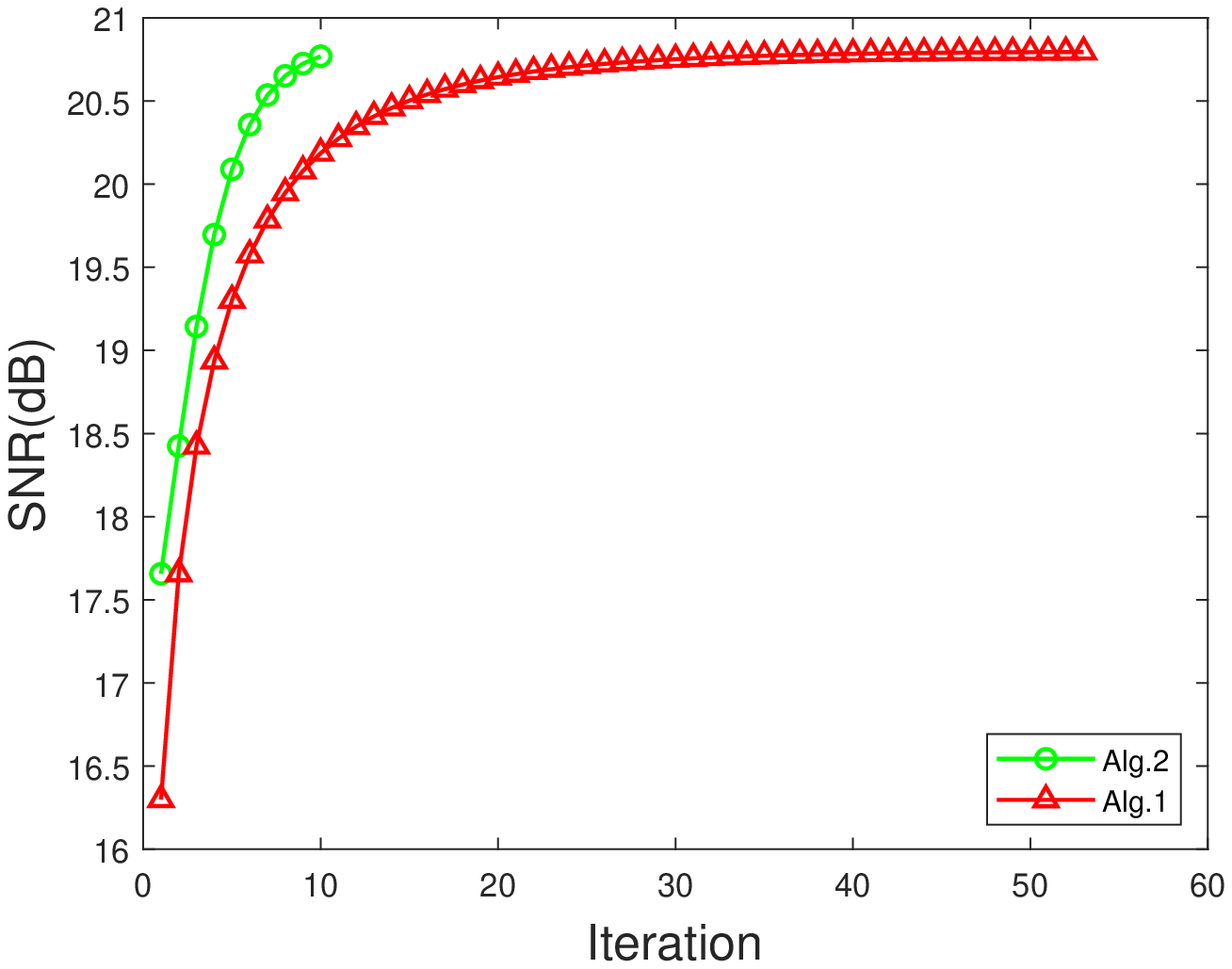}
\includegraphics[width=0.48\linewidth]{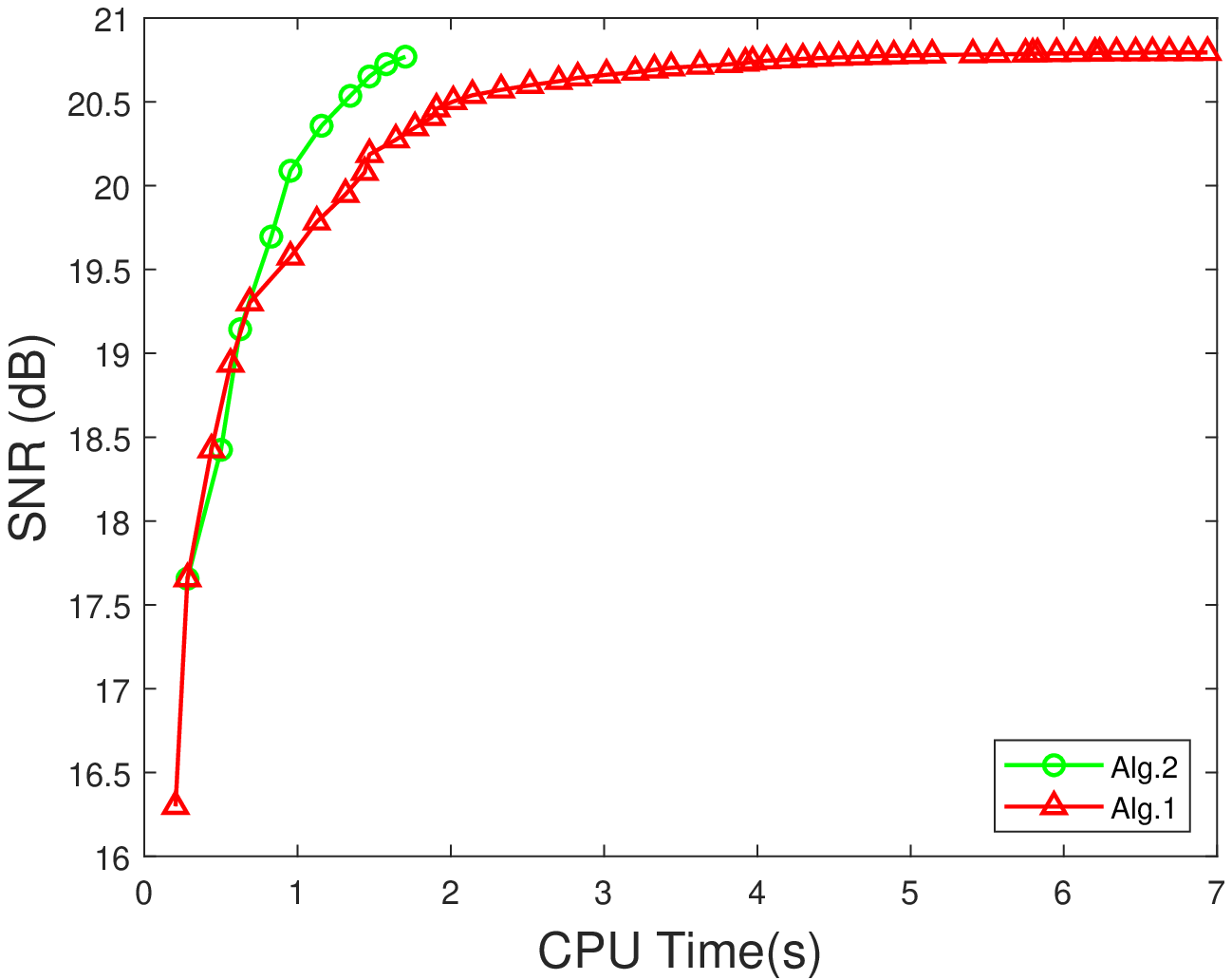}
\caption{Evolution of the SNR improvement with respect to iterations (left) and the CPU time (right) for Algorithms $2$ and $1$ (the classical AM method). The tested image is man and the blur kernel is the motion blur $(M,41,91)$. The noise level is $\sigma=10^{-3}$.}
\label{fig:am}
\end{figure}

\subsection{Comparison with FTVd and ADMM}
We compare the proposed methods with some existing solvers, e.g., FTVd \cite{yang-g}\footnote{\url{http://www.caam.rice.edu/~optimization/L1/ftvd/}} and ADMM-based methods \cite{salsa,csalsa,taomin,tval3,xiejx2018}, such as  (C)SALSA \cite{salsa,csalsa}\footnote{\url{cascais.lx.it.pt/~mafonso/salsa.html}}, TVAL3 \cite{tval3}\footnote{\url{http://www.caam.rice.edu/~optimization/L1/TVAL3/}}.
The efficiency of the ADMM-based methods and FTVd, compared to MFISTA \cite{beck-fast-ieee}, TwIST\cite{twist}, SpaRSA\cite{sparsa}, ALM \cite{liuzx2019,wuchunlin,wuchunlin2011}, LD \cite{chanld,tanzhao} and etc., have been verified by the authors of \cite{salsa,csalsa,taomin,tval3,yang-g}. Therefore, in this subsection we only present detailed numerical results comparing Algorithm $2$ to FTVd and ADMM-based methods (which we call ADMM in this paper).
Recall that FTVd is applied to solve the unconstrained problem \eqref{TV-f}, while ADMM tackles the following constrained problem
\begin{equation}
\label{ADMM_gmodel}
\min_{x,z}\bigg\{ \sum_{i=1}^{n^2}\|z\|_2+\frac{\mu}{2}\|\cK x-f\|^2_2 : z_i=D_ix,i=1,\ldots,n^2  \bigg\},
\end{equation}
which is equivalent to problem \eqref{TV} with $z$ being an auxiliary variable.

We use the codes provided by the authors of \cite{salsa,csalsa,tval3,yang-g} to implement the FTVd and ADMM (but with the stopping criterion \eqref{stop-rule}). Therefore, the values of all involved parameters of FTVd and ADMM remain unchanged.
During our implementation, we terminate all the algorithms with $\zeta=10^{-3}$ and all iterations start with the degraded images.

Tables \ref{table1} and \ref{table2} report the computing time in seconds and the restored SNR  of Algorithm $2$, FTVd and ADMM. The best method for each test case is highlighted in boldface. Where we use a randomized method similar to \cite{wuchunlin,yang-c} to generate the cross-channel blurs.
One can observe that Algorithm $2$ can recover images with the same quality as  FTVd and ADMM, but with fast speed.
Specifically, Algorithm $2$ is about $2$ times faster than FTVd for both gray and color images restoration. Interestingly, one can note that FTVd is better than ADMM for gray images, while ADMM is better than FTVd for color images.

We also compare the algorithms on Gaussian noise with different standard deviations, the blurring kernels and some detailed recovery results are given in Figures  \ref{fig:boat}, \ref{fig:man}, \ref{fig:mandrill} and \ref{fig:san}. We can see that Algorithm $2$ performs more competitively in restoring the same visible and SNR images. In addition, from Figure \ref{fig:mandrill}, one can find that ADMM is more efficient than both Algorithm $2$ and FTVd. In fact, during our implementation, we find that ADMM performs better for color images restoration with low deviation noise.

%

\begin{table}
\setlength{\tabcolsep}{2pt}
\caption{\label{table1} \footnotesize{Numerical comparison of Algorithm $2$ (Alg.2), FTVd and ADMM for images Boat and Man in Figure \ref{fig:test} (average of 10 runs). The noise level is $\sigma=10^{-3}$.} }
\centering
{\small
\begin{tabular}{ c c c c c c c c c c c   }
\hline

\hline  \\
\multicolumn{1}{ c }{Kernel} & &\multicolumn{1}{c }{Images} & &\multicolumn{3}{c }{Time (s)} & &\multicolumn{3}{c }{SNR (dB)}\\
 \cline{5-7} \cline{9-11}
\\[1pt]
& & & & Alg.2&FTVd&ADMM &&  Alg.2& FTVd & ADMM\\
\hline

\hline
\\
\multirow{1}{*}{\hbox{G$(11,9)$}}  & &Boat&&{\bf 0.65}&1.06&1.45& & 16.80&16.91&16.78 \\[2pt]
 & &Man& &{\bf 2.44}&3.84&6.94& & 18.95&19.03&18.86\\[2pt]
 \multirow{1}{*}{\hbox{G$(21,11)$}}  & &Boat&&{\bf 0.59}&1.33&1.75& & 12.92&13.01&12.98 \\[2pt]
 & &Man& &{\bf 2.30}&4.61&9.08& & 15.56&15.65&15.56\\[2pt]
  \multirow{1}{*}{\hbox{G$(31,13)$}}  & &Boat&&{\bf 0.63}&1.26&2.17& & 10.84&10.88&10.81 \\[2pt]
 & &Man& &{\bf 2.33}&4.79&10.61& & 13.72&13.81&13.76\\[2pt]
  \multirow{1}{*}{\hbox{M$(21,45)$}}  & &Boat&&{\bf 0.63}&0.92&1.17& &20.10&20.11&20.11 \\[2pt]
 & &Man& &{\bf 2.21}&3.67&4.71& & 22.58&22.59&20.57\\[2pt]
   \multirow{1}{*}{\hbox{M$(41,90)$}}  & &Boat&&{\bf 0.68}&1.01&1.18& &19.17&19.02&18.99 \\[2pt]
 & &Man& &{\bf 2.37}&3.98&5.67& &20.77&20.74&20.80\\[2pt]
   \multirow{1}{*}{\hbox{M$(61,135)$}}  & &Boat&&{\bf 0.54}&1.04&1.40& &15.83&16.01&15.98 \\[2pt]
 & &Man& &{\bf 2.07}&4.15&6.42& &19.14&19.24&19.14\\[2pt]
    \multirow{1}{*}{\hbox{A$(11)$}}  & &Boat&&{\bf 0.72}&1.07&1.40& &17.11&17.21&17.08 \\[2pt]
 & &Man& &{\bf 2.28}&3.72&6.92& &19.23&19.30&19.14\\[2pt]
     \multirow{1}{*}{\hbox{A$(13)$}}  & &Boat&&{\bf 0.66}&1.10&1.41& &16.31&16.41&16.30 \\[2pt]
 & &Man& &{\bf 2.26}&4.26&7.46& & 18.42&18.50&18.36\\[2pt]
    \multirow{1}{*}{\hbox{A$(15)$}}  & &Boat&&{\bf 0.60}&1.07&1.61& &15.51&15.62&15.55 \\[2pt]
 & &Man& &{\bf 2.29}&4.38&7.51& &17.75&17.83&17.70\\[2pt]
\hline

\hline
\end{tabular}
}
\end{table}

\begin{table}
\setlength{\tabcolsep}{2pt}
\caption{\label{table2} \footnotesize{Numerical comparison of Algorithm $2$ (Alg.2), FTVd and ADMM for images Mandrill and San Diego in Figure \ref{fig:test} (average of 10 runs). The noise level is $\sigma=10^{-3}$. }}
\centering
{\small
\begin{tabular}{ c c c c c c c c c c c  }
\hline

\hline  \\
\multicolumn{1}{ c }{Kernel} & &\multicolumn{1}{c }{Images} & &\multicolumn{3}{c }{Time (s)} & &\multicolumn{3}{c }{SNR (dB)}\\
 \cline{5-7} \cline{9-11}
\\[1pt]
& & & &Alg.2& FTVd&ADMM && Alg.2& FTVd& ADMM \\
\hline

\hline
\\
\multirow{1}{*}{\hbox{G$(11,9)$}}  & &Mandrill&&{\bf3.04}& 6.50&3.92&  & 11.39 & 11.31 & 11.30 \\[2pt]
 & &San Diego& &{\bf 12.61}&25.40&16.70& & 13.68 & 13.68 & 13.66\\[2pt]
\multirow{1}{*}{\hbox{G$(21,11)$}} & & Mandrill &  & {\bf3.09} &7.45 & 4.46 & & 8.58 & 8.53&8.51   \\[2pt]
 & &San Diego & & {\bf 12.99} & 27.67 & 18.60 & & 11.32 & 11.34 & 11.34 \\[2pt]
 \multirow{1}{*}{\hbox{G$(31,13)$}} & & Mandrill &  &{\bf 3.23} & 7.53 & 4.76 & & 7.48 & 7.46 & 7.44   \\[2pt]
 & &San Diego  && {\bf 12.74} & 27.04 & 19.79  & & 10.28 & 10.30 & 10.32 \\[2pt]
 \multirow{1}{*}{\hbox{M$(21,45)$}} & & Mandrill &  & {\bf 3.16} &  4.87 &3.43 & &16.74 & 16.78 &16.75   \\[2pt]
 & &San Diego  && {\bf 12.91}& 21.08 &12.10 & &18.99 & 19.04  &19.04\\[2pt]
 \multirow{1}{*}{\hbox{M$(41,90)$}} & & Mandrill&   & {\bf 3.09} &  5.72 &3.55 && 14.06 & 14.11 &14.08   \\[2pt]
 & &San Diego  && {\bf 13.02} & 24.16  &13.84 & & 16.58 &  16.64 &16.63\\[2pt]
 \multirow{1}{*}{\hbox{M$(61,135)$}} & & Mandrill&    & {\bf 3.05 } &6.27 & 4.00 && 12.20 & 12.19&12.16   \\[2pt]
 & &San Diego  && {\bf 13.00} &25.25 & 14.86 && 14.82 & 14.83 &14.82\\[2pt]
 \multirow{1}{*}{\hbox{A$(11)$}} & & Mandrill &  & {\bf 3.20 } &6.88 & 4.02&& 11.77 & 11.70&11.68   \\[2pt]
 & &San Diego   && {\bf 12.88} &25.30 & 16.82 & & 14.01 & 14.01 &13.99\\[2pt]
 \multirow{1}{*}{\hbox{A$(13)$}} & & Mandrill&  & {\bf 3.07} &6.68 & 4.36 && 10.95 & 10.87&10.86   \\[2pt]
 & &San Diego  && {\bf 13.00} &25.52 &16.29 & & 13.35 & 13.35 &13.33\\[2pt]
 \multirow{1}{*}{\hbox{A$(15)$}} & &Mandrill &   & {\bf 3.22} &  7.37 & 4.36 & & 10.32 & 10.25&10.24   \\[2pt]
 & &San Diego  & &{\bf 13.29} &25.91 &16.77 & & 12.84 & 12.85 &12.84\\[2pt]
 \multirow{1}{*}{Cross-channel} & &Mandrill &   & {\bf 3.22} &6.65 & 4.16 & & 11.07 & 11.13&11.13   \\[2pt]
 & &San Diego  && {\bf 13.55} &26.29 & 15.74 & & 13.89 & 14.01 &14.05\\[2pt]
\hline

\hline
\end{tabular}
}
\end{table}

\begin{figure}
\centering
\includegraphics[width=0.95\linewidth]{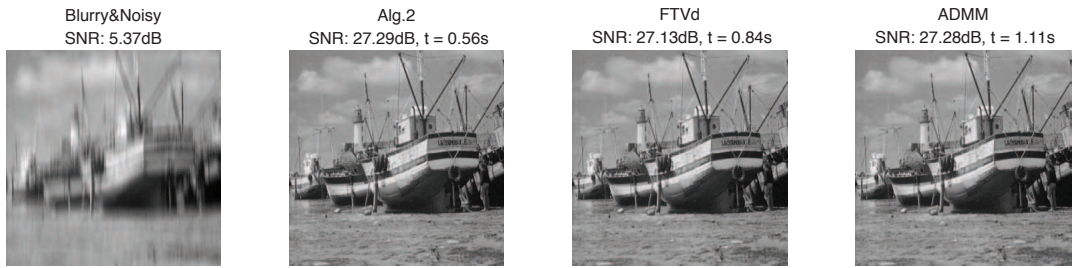}
\caption{\footnotesize{Comparisons between Algorithm $2$ (Alg.2),  FTVd and ADMM for the image boat. The blur kernel is motion blur $(M,41,90)$. The noise level is $\sigma=10^{-4}$.}}
\label{fig:boat}
\end{figure}

\begin{figure}
\centering
\includegraphics[width=0.95\linewidth]{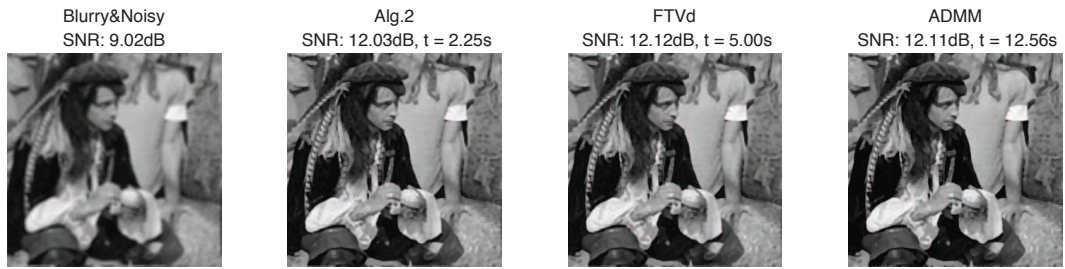}
\caption{ \footnotesize{Comparisons between Algorithm $2$ (Alg.2),  FTVd and ADMM for the image man. The blur kernel is Gaussian blur $(G,21,11)$. The noise level is $\sigma=10^{-2}$.}}
\label{fig:man}
\end{figure}

\begin{figure}
\centering
\includegraphics[width=0.95\linewidth]{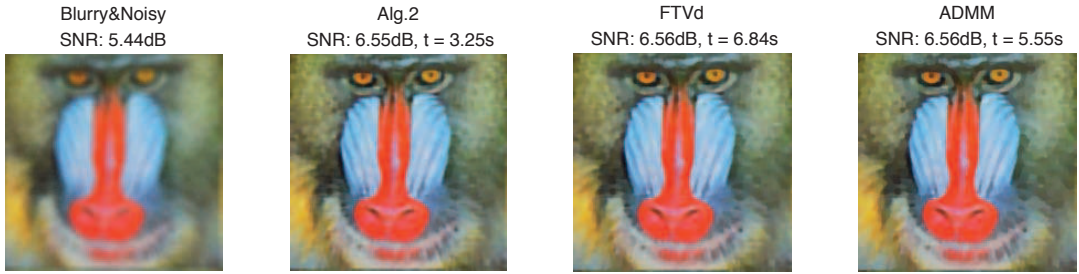}
\caption{Comparisons between  Algorithm $2$ (Alg.2),  FTVd and ADMM for the image mandrill. The blur kernel is (M, 41,90) and the noise level is $\sigma=10^{-2}$. }
\label{fig:mandrill}
\end{figure}

\begin{figure}
\centering
\includegraphics[width=0.95\linewidth]{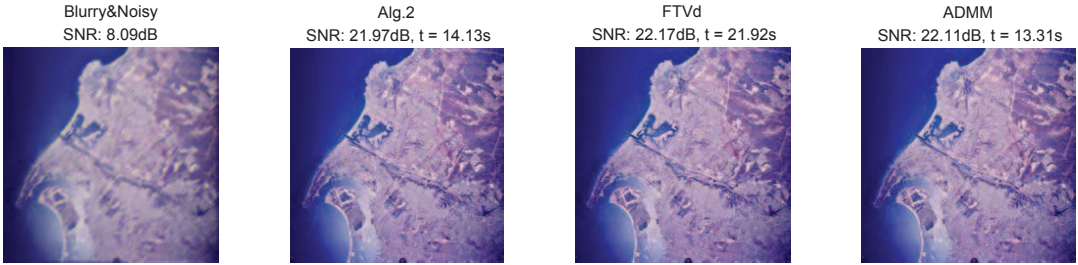}
\caption{Comparisons between  Algorithm $2$ (Alg.2),  FTVd and ADMM for the image San Diego. The blur kernel is cross-channel blur and the noise level is $\sigma=10^{-4}$.}
\label{fig:san}
\end{figure}

\subsection{Comparison with other existing solvers}
In the literature, there are some other efficient solvers applicable to the model \eqref{TV}. In order to conduct a more comprehensive comparison, we consider a more general model derived for the sparse signal recovery problem. In this situation, the linear operator $\mathcal{K}$ in \eqref{TV} will be recognized as a certain measurement matrix and $D$ is a analysis operator such that $Dx$ is a sparse vector. Such an problem arises from  a very active field of recent research named compressed sensing. We refer to \cite{caijf2018,csbook,yall1} for more discussions.

In our test, the measurement matrix $\mathcal{K}$ will be chosen as normalized Gaussian matrices, i.e., its entries are generated from $i.i.d.$ normal distributions $\mathcal{N}(0,1)$ (\texttt{randn(m,n)} in \textsc{Matlab}) with columns being normalized. We use two kinds of the analysis operator:
(\romannumeral1) random tight frames with $D^TD=I$, where $D$ comprises the first $n$ columns of $Q$ obtained by QR factorization on a $p\times n$ Gaussian matrix ;
(\romannumeral2) $n\times n$ discrete cosine transform (DCT) matrices.
The original vector $x$ is obtained by $x=D^Ty$ with $y$ being a $s$-sparse vector (the number of nonzero entries of a certain vector is less than or equal to $s$) and the noisy vector $\eta\in \mathbb{R}^m$ is generated from $\mathcal{N}(0,1)$. After this, we set $f:=\mathcal{K} x+ \sigma\eta$ with $\sigma=10^{-3}$.

We compare Algorithm $2$ with the existing solvers including NESTA \cite{nesta} (vision $1.1$)\footnote{\url{http://statweb.stanford.edu/~candes/nesta/}}, YALL1 \cite{yall1} (vision $1.4$)\footnote{\url{http://yall1.blogs.rice.edu/}} and SFISTA \cite{tanzhao}, which had been shown to be favorable among other algorithms such as the interior point method (e.g., $\ell_1$-$\ell_s$)\cite{l1ls},  the nonlinear conjugate gradient descend (CGD) \cite{CGD} algorithm, the generalized iterative soft-thresholding (GIST) \cite{GIST} algorithm, etc.

We use the codes provided by the authors of \cite{nesta,yall1} to implement the NESTA and YALL1 (but with the stopping criterion \eqref{stop-rule}). The values of all involved parameters of NESTA and YALL1 remain unchanged. The codes for SFISTA is coded by us. We set $\beta:=2^{11}$ for Algorithm $2$. All the algorithms are terminated if \eqref{stop-rule} is satisfied with $\zeta:=10^{-6}$.

Tables \ref{table3} and \ref{table4} list the numerical comparison of these algorithms, where ``Problem'' denotes the problem size $m\times n$, ``Time'' denotes the average CPU time (in seconds) of the average $10$ runs, ``Error'' denotes the average relative error, respectively.
Data in these tables show that Algorithm $2$ is faster than the other solvers to find a solution that has almost the same relative error.
Specifically, Algorithm $2$ is about $2$-$5$ times faster than NESTA and is almost $100$ times faster than SFISTA.

\begin{table}
\setlength{\tabcolsep}{2pt}
\caption{\label{table3} \footnotesize{Comparisons between NESTA, YALL1, SFISTA and Algorithm $2$ (Alg.2)  with random tight frames (average of 10 Runs). }}
\centering
{
\begin{tabular}{ l c c c c c c c c c c  }
\hline

\hline  \\
\multicolumn{1}{ l }{Problem} & &\multicolumn{4}{c }{Time (s)} & &\multicolumn{4}{c }{Error}\\
\cline{1-1} \cline{3-6} \cline{8-11}
\\[1pt]
\hbox{$m/n$} & & NESTA & YALL1 & SFISTA & Alg.2 & & NESTA & YALL1 & SFISTA & Alg.2\\
\hline

\hline
\\
$256/1024$ & &0.63 &0.37 & 33.25 & {\bf 0.27} & & 8.36e-3 & 2.98e-3 & 3.27e-3 & 3.68e-3 \\[2pt]
$256/2048$ & &5.71 &3.44 & 359.93 & {\bf 1.55} & & 1.49e-2 & 4.08e-3 & 5.31e-3 & 6.54e-3 \\[2pt]
$256/4096$ & &40.46 &25.59 & 2033.23 & {\bf 8.13} & & 2.23e-2 & 4.44e-3 & 7.22e-3 & 0.99e-3 \\[2pt]
$256/8192$ & &282.94 &170.21 & 8691.36 & {\bf 45.15} & & 4.49e-2 & 5.98e-3 & 5.31e-3 & 2.14e-3 \\[2pt]
$512/2048$ & &4.21 &2.19 & 242.31 & {\bf 1.54} & & 8.99e-3 & 3.11e-3 & 3.02e-3 & 3.87e-3 \\[2pt]
$512/4096$ & &26.30 &15.91 & 1213.60 & {\bf 6.60} & & 1.49e-2 & 4.10e-3 & 2.62e-3 & 6.49e-3 \\[2pt]
$512/8192$ & &159.33 &95.20 & 7825.55 & {\bf 30.76} & & 2.48e-3 & 4.85e-3 & 7.88e-3 & 1.12e-2 \\[2pt]
$1024/2048$ & &2.76 &2.17 & 108.56 & {\bf 1.68} & & 6.67e-3 & 2.62e-3 & 2.75e-3 & 2.91e-3 \\[2pt]
$1024/4096$ & &16.98 &8.90 & 930.31 & {\bf 6.03} & & 9.26e-3 & 3.12e-3 & 3.38e-3 & 3.87e-3 \\[2pt]
$1024/8192$ & &100.02 &61.59 & 5859.97 & {\bf 24.88} & & 1.41e-2 & 3.74e-3 & 5.68e-3 & 5.98e-3 \\[2pt]
\hline

\hline
\end{tabular}
}
\end{table}


\begin{table}
\setlength{\tabcolsep}{2pt}
\caption{\label{table4} \footnotesize{Comparisons between NESTA, YALL1, SFISTA and Algorithm $2$ (Alg.2)  with with DCT matrices (average of 10 Runs). }}
\centering
{
\begin{tabular}{ l c c c c c c c c c c  }
\hline

\hline  \\
\multicolumn{1}{ l }{Problem} & &\multicolumn{4}{c }{Time (s)} & &\multicolumn{4}{c }{Error}\\
\cline{1-1} \cline{3-6} \cline{8-11}
\\[1pt]
\hbox{$m/n$} & & NESTA & YALL1 & SFISTA & Alg.2 & & NESTA & YALL1 & SFISTA & Alg.2\\
\hline

\hline
\\
$256/1024$ & &0.48 &0.18 & 20.60 & {\bf 0.15} & & 8.77e-2 & 3.11e-3 & 3.48e-3 & 3.89e-3 \\[2pt]
$256/2048$ & &1.05 &0.44 & 42.44 & {\bf 0.23} & & 1.48e-2 & 4.12e-3 & 5.31e-3 & 6.52e-3 \\[2pt]
$256/4096$ & &4.13 &1.70 & 108.60 & {\bf 0.64} & & 2.47e-2 & 4.97e-3 & 8.59e-3 & 1.13e-2 \\[2pt]
$256/8192$ & &14.92 &6.37 & 277.25 & {\bf 1.75} & & 3.88e-2 & 4.69e-3 & 1.30e-2 & 1.80e-2 \\[2pt]
$512/2048$ & &1.20 &0.50 & 44.29 & {\bf 0.47} & & 8.89e-3 & 3.00e-3 & 3.06e-3 & 3.77e-3 \\[2pt]
$512/4096$ & &4.93 &2.05 & 160.44 & {\bf 1.01} & & 1.43e-2 & 3.08e-3 & 4.31e-3 & 6.13e-3 \\[2pt]
$512/8192$ & &17.19 &7.48 & 409.21 & {\bf 2.60} & & 2.37e-3 & 4.65e-3 & 9.18e-3 & 1.06e-3 \\[2pt]
$1024/2048$ & &1.36 &{\bf 0.78} & 33.07& 0.89 & & 6.76e-3 & 2.66e-3 & 2.88e-3 & 2.95e-3 \\[2pt]
$1024/4096$ & &5.19 &2.06 & 163.55 & {\bf 1.79} & & 9.64e-3 & 3.23e-3 &3.08e-3 & 4.03e-3 \\[2pt]
$1024/8192$ & &19.17 &8.22 & 563.51 & {\bf 3.79} & & 1.42e-2 & 3.78e-3 & 4.83e-3 & 6.04e-3 \\[2pt]
\hline

\hline
\end{tabular}
}
\end{table}

%
%
%
%
%
%
%

\section{Conclusions}
\label{sect.7}
In this paper, we proposed a new symmetric AM algorithm for total variation minimization.
The proposed algorithm can not only keep the computational simplicity of AM, but also share a fast convergence rate.
Convergence of the proposed algorithm is established under the equivalence built between the proposed algorithm and the APG method.
Numerical results, including comparison with some popular solvers, shown that our algorithm is very efficient.
We believe that the proposed algorithm can be extended to a number of models involving TV regularization, such as TV-L1 \cite{taomin,xu2014,yang-l1}, TV-based Poisson noise removal \cite{zhang-2018}.

\end{document}